\documentclass[11pt]{article}

\usepackage[letterpaper, portrait, margin=1in]{geometry}

\usepackage[usenames,dvipsnames,svgnames,table]{xcolor}

\newcommand*\ie{i.\kern.1em e.\ }
\newcommand*\eg{e.\kern.1em g.\ }

\usepackage[utf8]{inputenc}

\usepackage[normalem]{ulem}
\usepackage{mathtools,bm}
\usepackage{amsthm}
\usepackage{thmtools}
\usepackage[inline]{enumitem}
\usepackage{thm-restate}
\usepackage{amsmath,amsfonts,amssymb,eqparbox,array,ragged2e,stmaryrd,bbm}
\usepackage{physics}
\usepackage{xspace}
\usepackage{tikz}
\usepackage{ifthen}
\usepackage{pgfmath-xfp}
\usetikzlibrary{positioning,decorations.pathreplacing,calligraphy, arrows.meta}
\pgfmathsetseed{3434}
\usepackage{algorithm}
\usepackage{algpseudocode}

\usepackage[colorlinks=true,linkcolor=RawSienna,citecolor=ForestGreen]{hyperref}

\usepackage{thm-restate}
\usepackage{enumitem}
\usepackage[capitalise,nameinlink,noabbrev]{cleveref}

\usepackage{subfig}
\newsavebox{\tempbox}

\usepackage{tikz}

\usetikzlibrary{
    babel,
    arrows,
	calc,
	patterns,
	intersections,
    backgrounds,
    trees,
    mindmap,
    fadings,
	decorations.pathreplacing,
	decorations.pathmorphing,
	decorations.text,
	decorations.markings,
    scopes,
	shapes,
    snakes,
    shapes,
	positioning,
    arrows.meta,
    patterns,
    patterns.meta,
    angles,
    quotes,
    math
}

\tikzfading[name = middle,
    top color = transparent!100,
    bottom color = transparent!100,
    middle color = transparent!30
]

\tikzset{
    >=latex,
    pics/rhom/.style n args = {4}{
        code = {
            \draw[
                thick, #3, fill = #3!50, fill opacity = 0.2, rounded corners = 3pt]
                (-#1, 0) -- (0, -#2) -- (#1, 0) -- (0, #2) -- cycle;
            \node[draw, fill, circle, inner sep = 0pt, minimum size = 0.05cm, #3] (#4) at (0, 0) {};
        }
    },
    linecol/.style n args = {3}{
        postaction = {
            decorate,
            decoration = {
                markings,
                mark = between positions 0 and 0.9 step 0.2pt with {
                    \pgfmathsetmacro\myval{
                        multiply(
                            divide(
                                \pgfkeysvalueof{/pgf/decoration/mark info/distance from start},
                                \pgfdecoratedpathlength
                            ),
                            120
                        )};
                    \pgfsetfillcolor{#3!\myval!#2};
                    \pgfpathcircle{\pgfpointorigin}{#1};
                    \pgfusepath{fill};
                },
                mark = at position 1 with {\arrow[#3]{latex}},
            }
        }
    },
    pics/universe-rect/.style n args = {2}{
        code = {
            \fill[black!10, rounded corners = 2pt] (-0.07, -0.07) rectangle (#1 + 0.07, #2 + 0.07);
            \draw[fill = white, thick] (0, 0) rectangle (#1, #2);
        }
    },
    pics/comm-rect/.style n args = {4}{
        code = {
            \draw[
                thick, #3, fill = #3!50, fill opacity = 0.2, rounded corners = 3pt]
                (-#1, -#2) rectangle (#1, #2);
            \node[draw, fill, circle, inner sep = 0pt, minimum size = 0.05cm, #3] (#4) at (0, 0) {};
        }
    },
    fancy-arrow/.style = {
        draw = black,
        thick,
        single arrow,
        right color = VioletRed!50,
        left color = ForestGreen!50,
        fill opacity = 0.7,
        single arrow head extend = 0.3cm,
        single arrow tip angle = 70,
        single arrow head indent = 0.15cm,
        minimum height = 2.4cm,
        minimum width = 0.8cm,
        shading angle = -45,
        rotate = -90
    }
}

\algnewcommand\algorithmicinput{\textbf{Input: }}
\algnewcommand\Input{\item[\algorithmicinput]}
\algnewcommand\algorithmicoutput{\textbf{Output: }}
\algnewcommand\Output{\item[\algorithmicoutput]}
\algnewcommand{\OneLineIf}[2]{
  \State \algorithmicif\ #1\ \algorithmicthen\ #2}

\newcommand{\ignore}[1]{}

\declaretheorem[name=Theorem]{theorem}
\declaretheorem[sibling=theorem,name=Lemma]{lemma}

\declaretheorem[sibling=theorem,name=Claim]{claim}

\declaretheorem[sibling=theorem,name=Fact]{fact}

\theoremstyle{definition}
\declaretheorem[sibling=theorem,name=Definition]{definition}

\declaretheorem[sibling=theorem,name=Question]{question}

\theoremstyle{remark}

\makeatletter
\def\th@example{%
  \thm@notefont{}
  \normalfont 
}
\def\th@definition{%
  \thm@notefont{}
  \normalfont 
}
\makeatother

\theoremstyle{example}

\renewcommand{\setminus}{\smallsetminus}
\renewcommand{\hat}{\widehat}

\renewcommand{\epsilon}{\varepsilon}
\newcommand{\eps}{\varepsilon}

\newcommand{\poly}{\mathrm{poly}}
\newcommand{\polylog}{\mathrm{polylog}}
\newcommand{\dTV}{d_{\mathrm{TV}}}

\newcommand{\Inf}{\mathrm{Inf}}
\newcommand{\adv}{\mathrm{corr}}
\newcommand{\dist}{\mathrm{dist}}
\newcommand{\avg}{\mathrm{avg}}

\newcommand{\newmeasure}[2]{\newcommand{#1}{{\textup{\sffamily #2}}\xspace}}
\newmeasure{\C}{C}

\ifthenelse{1>0}{

}

\allowdisplaybreaks
\begin{document}
\title{Quantum Advantage in Tolerant Junta Testing}

\author{Avishay Tal\thanks{University of California, Berkeley. Email: \href{mailto:atal@berkeley.edu}{atal@berkeley.edu}. Supported by an NSF CAREER Award
CCF-2145474.}
\and 
Weiqiang Yuan\thanks{EPFL. Email: \href{mailto:weiqiang.yuan@epfl.ch}{weiqiang.yuan@epfl.ch}. Supported by the Swiss State Secretariat for Education, Research and Innovation (SERI) under contract number MB22.00026.}}



\maketitle

\begin{abstract}
We establish the first super-polynomial quantum advantage for the tolerant junta testing problem in the adaptive setting. Specifically, we show that within a certain parameter regime, tolerant $k$-junta testing with high precision can be solved using $\mathrm{poly}(k)$ quantum queries, whereas any classical algorithm requires at least $k^{\Omega(\log k)}$ queries.

The problem of tolerant $k$-junta testing is as follows: given parameters $(k, \epsilon_1, \epsilon_2)$, with $0\le \epsilon_1<\epsilon_2 \le 1/2$, and black-box access to a Boolean function $f$ (defined on $n$ variables), distinguish whether $f$ is $\epsilon_1$-close to some $k$-junta or $\epsilon_2$-far from every $k$-junta. 

We show the quantum advantage for a range of parameters close to $1/2$, for example, $\epsilon_1 = 1/2-1/k$  and $\epsilon_2 = 1/2-1/(2k^2)$. (As such, the problem is more naturally captured using the notion of correlation with closest $k$-junta.)
The (non-adaptive) quantum tester we use was given by a recent work of Bao, Liu, Yao, Ye, and Zhang (SOSA 2026). We slightly adapt their analysis to show that it holds in the above parameter regime. On the other hand, our classical lower bound requires substantial new ideas. Inspired by the lower bound techniques of Chen and Patel (FOCS 2023), we introduce a new hard distribution of ``yes'' instances (i.e., instances with distance at most $\epsilon_1$ to $k$-juntas) that is based on planting an ``approximate-junta'' as follows: we randomly pick $k$ out of $n$ coordinates, and for each fixing of the $k$ coordinates, the $2^{n-k}$ values in the restricted subcube are drawn randomly except for the set of points in an error-correcting code on which we place the same random bit. We show that this distribution is much closer to $k$-juntas than the uniform distribution, but on the other hand, they are indistinguishable with respect to any classical algorithm making $k^{o(\log k)}$ queries.

\end{abstract}

\vspace{2em}

\thispagestyle{empty}
\setcounter{page}{0}


\newpage
\setcounter{page}{1}

\section{Introduction}

A Boolean function $f:\{0,1\}^n\to \{\pm 1\}$ is said to be a $k$-junta if it depends on at most $k$ of $n$ input variables.
Despite representing a structurally simple class of functions, juntas play a foundational role across various areas of theoretical computer science, including computational learning theory~\cite{MOS03,FGKP06,Val15}, 
hardness amplification~\cite{BHKT24,Kumar26}, 
and property testing (see references below), which is the focus of this study.

In the standard junta testing problem, we are given oracle access to a hidden Boolean function $f:\{0,1\}^n\to \{\pm 1\}$.
The objective is to determine whether $f$ is a $k$-junta or $\epsilon$-far from all $k$-juntas.
To be precise, a \emph{randomized} algorithm $\mathcal{A}^f$ is a $(k,\epsilon)$-junta tester if it accepts every $k$-junta $f$ with probability at least $2/3$, while rejecting with probability at least $2/3$ every $f$ that differs from every $k$-junta on at least an $\epsilon$-fraction of inputs.
We measure the tester's efficiency by its query complexity, namely, the number of oracle calls to $f$ (while the running time is not considered).
Ideally, we want the tester's query complexity to be as small as possible in terms of $k$ and $\epsilon$, but independent of $n$ (the input length of $f$).

In the seminal work of Fischer, Kindler, Ron, Safra, and Samorodnitsky~\cite{testingjuntas}, the authors proposed the first junta tester with query complexity depending only polynomially on $k$ and $1/\epsilon$. Since then, a long line of research~\cite{CG04, HK07, atici2007quantum, 
Blais08, Blais09, STW15, ABRdW16, Sag18, CSTWX18, junta-submodular-opt,LCSSX18, LW19,Belovs19,Bsh19,DMN19,ITW21,PRW22,CP23, CNY23, NP24,CDLNS24,CPS26,CPPS25}
has been dedicated to developing more efficient testers and sharper lower bounds, leading to a comprehensive understanding of the problem and its variants.
In the standard model, Blais~\cite{Blais09} provided a tester using $O(k\log k + k/\epsilon)$ queries, which was later shown to be optimal for constant $\epsilon$~\cite{Sag18}.
On the other hand, ~\cite{CSTWX18} shows that any non-adaptive tester must make $\tilde{\Omega}(k^{3/2}/\epsilon)$ queries, nearly matching the $\tilde{O}(k^{3/2}/\epsilon)$ upper bound established by Blais~\cite{Blais08}.
Here, a tester is non-adaptive if its queries do not depend on answers to previous queries; otherwise it is adaptive.

Beyond the classical setting, quantum algorithms have been shown to achieve a quadratic speedup: ~\cite{ABRdW16} provides a quantum tester making only $\tilde{O}(\sqrt{k/\epsilon})$ queries, which is nearly optimal~\cite{BKT20}.
Lastly, recent work has extended these results to more general models, such as distribution-free testing~\cite{LCSSX18,Bsh19} and relative-error testing~\cite{CPPS25}, demonstrating that junta testing remains tractable even under these broader distance measures.
However, this efficiency heavily relies on the assumption of perfect completeness.

\paragraph{Tolerant Testing.} A significant limitation of the standard testing model is that there is no guarantee on the performance of the tester on functions close to satisfying the property. 
A more realistic scenario is that given a black-box function $f$, one seeks to test whether there are $k$ features that explain the output \emph{in most cases}, or $f$ is far from it. This is natural because data in the wild are often noisy and imperfect.
To address this, Parnas, Ron, and Rubinfeld~\cite{PRR06} introduced a natural relaxation of the standard testing model called \emph{tolerant testing}.

In the context of junta testing, a $(k,\epsilon_1,\epsilon_2)$-tolerant junta tester must decide whether a Boolean function $f$ is $\epsilon_1$-close to some $k$-junta or $\epsilon_2$-far from every $k$-junta (for $\epsilon_1<\epsilon_2$) by making oracle calls to $f$.
Standard junta testing naturally emerges as a special case where $\epsilon_1 = 0$.

In sharp contrast to the standard model, tolerant junta testing appears to be much harder. In the non-adaptive setting, Chen et al.~\cite{CDLNS24} proved that any classical tolerant junta tester requires $2^{\Omega(\sqrt{k})}$ queries. Subsequently, this lower bound was shown to be tight by Nadimpalli and Patel~\cite{NP24} (building on~\cite{DMN19, ITW21}), who provided a non-adaptive upper bound matching $2^{\tilde{O}(\sqrt{k\log (1/(\epsilon_2-\epsilon_1))})}$.

In the classical adaptive setting, despite significant effort, our understanding of tolerant junta testing remains limited.
The best-known adaptive $(k,\epsilon_1,\epsilon_2)$-tolerant junta tester has query complexity $2^{\tilde{O}(k^{1/3} \polylog(1/(\epsilon_2-\epsilon_1)))}$~\cite{CPS26}, while the best-known lower bound is $k^{\Omega(\log(1/(\epsilon_2-\epsilon_1)))}$ when $\epsilon_1,\epsilon_2$ are bounded away from $0$ and $1/2$~\cite{CP23}.
In short, even in the parameter regime where $\epsilon_1=0.01$ and $\epsilon_2=0.49$, neither a tester with query complexity $\exp(o(k^{1/3}))$ nor a superpolynomial lower bound is known.

In the quantum world, however, the landscape changes dramatically.
Recent work by Bao, Liu, Yao, Ye, and Zhang~\cite{quantumtester} introduces an elegant, non-adaptive quantum tester with query complexity $O(k\log k)$, provided that $4\epsilon_1<\epsilon_2$.
To evaluate this advantage, Bao et al.~\cite{quantumtester} complement their algorithm with a classical non-adaptive lower bound,
demonstrating a surprising exponential quantum speedup over classical \emph{non-adaptive} tolerant junta testers.

This provable advantage highlights the power of quantum in tolerant junta testing.
However, \cite{quantumtester}'s results do not establish any super-polynomial quantum speedup over classical \emph{adaptive} tolerant junta testers, since the quantum tester in~\cite{quantumtester} works only in a parameter regime that is mutually exclusive of that of the superpolynomial lower bound by \cite{CP23}.%
\footnote{More precisely, the quantum tester is efficient if either the gap $\epsilon_2-\epsilon_1$ is constant, or both $\epsilon_1,\epsilon_2$ are close to $0$, while the only known super-polynomial classical adaptive lower bound~\cite{CP23} requires a sub-constant gap between $\epsilon_1,\epsilon_2$ and that both $\epsilon_1,\epsilon_2$ are bounded away from $0$ and $1/2$, which we will discuss in more detail in~\Cref{sec: overview classical lower bound}.}
This leads to the following natural question:
\begin{quote}
\emph{Do adaptive quantum tolerant junta testers significantly outperform their classical counterparts?}
\end{quote}

\subsection{Main Result}
In this paper, we give an affirmative answer to the above question by showing a super-polynomial quantum speedup against classical adaptive tolerant junta testers.
\begin{theorem}[Main Result]\label{theorem: main}
For every $k\in \mathbb{N}$ and $\epsilon>2^{-k^{.99}}$, there exist $0<\epsilon_1<\epsilon_2<1/2$ such that $\epsilon = \epsilon_2-\epsilon_1$ and the following holds:
\begin{itemize}[noitemsep]
\item There exists a {\sf non-adaptive quantum} $(k,\epsilon_1,\epsilon_2)$-tolerant junta tester with query complexity $O(\poly(k,1/\epsilon))$.
\item Any {\sf adaptive classical} $(k,\epsilon_1,\epsilon_2)$-tolerant junta tester must make at least $k^{\Omega(\log (1/\epsilon))}$ queries.
\end{itemize}
\end{theorem}

The key to the proof of~\Cref{theorem: main} is a stronger classical adaptive lower bound, inspired by~\cite{CP23}, which applies to a wider range of parameters for which the quantum tester in~\cite{quantumtester} is valid.
The quantum advantage is perhaps better captured by the notion of {\sf correlation} with the closest $k$-junta, which is equivalent to the distance up to a linear transformation.
In this notation, we show that for any $\tau>2^{-k^{.99}}$, 
\begin{description}
    \item[Quantum Tester (\Cref{theorem: quantum tester})]There exists a non-adaptive quantum algorithm, making $\poly(k/\tau)$ queries, that distinguishes between functions with at least $\tau_1=\tau$ correlation and at most $\tau_2=\tau^2/2$ correlation to $k$-juntas.
    \item[Classical Lower Bound (\Cref{theorem: classical lower bound})] On the other hand, we show that distinguishing $\tau_1= \tau$ from $\tau_2= \tau^2/2$ correlation with $k$-juntas (and even $\tau_1 = \tau$ versus $\tau_2 = \tau^c$ for any constant $c>1$) requires $k^{\Omega(\log(1/\tau))}$ classical queries, even adaptively.
\end{description}
By picking $\tau = 1/\poly(k)$, we obtain a quantum algorithm with $\poly(k)$ queries, compared to a quasi-polynomial classical lower bound of $k^{\Omega(\log k)}$ queries.


Prior to our work, super-polynomial classical lower bounds for tolerant junta testing were only known for much closer parameters $\tau_1,\tau_2$  with $\tau_1= \tau_2\cdot (1+o(1))$~\cite{CP23}. 

We view this as a quantum advantage to a natural tolerant property testing problem: distinguish whether there exist $k$ features that explain at least $\tau$ or at most $\tau^2/2$ fraction of the outputs.

\subsection{Proof Overview}

\subsubsection{Quantum Tester}

Our quantum tolerant junta tester is a simple adaptation of the one provided in~\cite{quantumtester}.
\begin{theorem}[{\cite[Theorem~1.3]{quantumtester}}]
 Let $0<\epsilon_1<\epsilon_2<1/2$ be constants such that $4\epsilon_1<\epsilon_2$.
Then, there exists a non-adaptive quantum $(k,\epsilon_1,\epsilon_2)$-junta tester with query complexity $O(k \log k)$.
\end{theorem}

To demonstrate this, Bao et al.~\cite{quantumtester} first approximate the minimum distance of a Boolean function $f$ and $T$-juntas by $\Inf_{\overline{T}}[f]$, where $\Inf_{\overline{T}}(f)$ is expressed in terms of the Fourier expansion of $f$ as $\sum_{S\subseteq[n]: S \nsubseteq T }\hat{f}_S^2$.
Specifically, they show
\[
\Inf_{\overline{T}}[f]/4\le \dist(f,\mathcal{J}_T)\le \Inf_{\overline{T}}[f],
\]
where $\dist(f,\mathcal{J}_T)$ denotes the minimum distance of a Boolean function $f$ from $\mathcal{J}_T$, the class of juntas on $T$.
Then, they design an  ``influence tester'' that distinguishes between the following two cases:
\begin{itemize}
    \item Yes: $\Inf_{\overline{T}}[f]\le \gamma_1$ for some $T\subseteq [n]$ with size $|T|=k$
    \item No: $\Inf_{\overline{T}}[f]\ge \gamma_2$ for all $T\subseteq [n]$ with size $|T|=k$
\end{itemize}
for any $\gamma_1<\gamma_2$ with query complexity $\mathrm{poly}(k,(\gamma_2-\gamma_1)^{-1})$.

The quantum tester in \Cref{theorem: main} follows directly from combining their influence tester with the following improved characterization of $\adv(f,\mathcal{J}_T)$ in terms of $\inf_{\overline{T}}[f]$:
\[
1-\Inf_{\overline{T}}(f)\le \adv(f,\mathcal{J}_T)\le \sqrt{1-\Inf_{\overline{T}}(f)}.
\]
We emphasize that this characterization is not new: the upper bound was shown in~\cite[Lemma 29]{quantumtester} as an intermediate step, while the lower bound was provided in~\cite{testingjuntas}.

In what follows, we briefly discuss the main idea of the quantum influence tester in~\cite{quantumtester}.
Its efficiency stems from the fact that a single quantum query suffices to sample $S \subseteq [n]$ from the Fourier distribution of $f$. With a sufficient number of such samples, we are able to estimate $\Inf_{\overline{T}}[f]$ for any $T\subseteq [n]$ up to some small error. We emphasize here that while the query complexity is $\poly(k)$, the running time of the algorithm is exponential in $k$—although this is not the focus of this paper.

As our goal is to determine if $\min_{T:|T|=k} \Inf_{\overline{T}}[f]$ is smaller than $\gamma_1$ or greater than $\gamma_2$, we may want a good approximation of $\Inf_{\overline{T}}[f]$ for all $T$ with $|T|=k$.
However, since there are $\binom{n}{k}$ choices of subsets of $[n]$ with size $k$,
a naive union bound over all these subsets will incur a $\log n$ factor to the number of queries, which is intolerable.

De, Mossel, and Neeman~\cite{DMN19} provide a solution by showing that the tolerant junta testing problem for functions with an arbitrarily large input length $n$ can be reduced to the case where $n=\poly(k)$.
Subsequently, Iyer, Tal, and Whitmeyer~\cite{ITW21} further reduce $n$ to $O(k)$, while barely changing the distance to the closest $k$-junta. Bao et al.~\cite{quantumtester} use Fourier samples to get an even simpler and more efficient reduction, using only $O(k \log k)$ quantum queries.




\subsubsection{Classical Lower Bound}\label{sec: overview classical lower bound}

Our classical lower bound is inspired by the work of Chen and Patel~\cite{CP23},
where they give the first super-polynomial adaptive lower bound for tolerant junta testing.
Specifically, they prove the following.
\begin{theorem}[Theorem 1 in~\cite{CP23}]
    Let $0.01\le \epsilon_1<\epsilon_2\le 0.49$ be two parameters such that $\epsilon_2-\epsilon_1\ge 2^{-O(k^{.99})}$.
    Then any $(k,\epsilon_1,\epsilon_2)$-tolerant junta tester requires $k^{\Omega(\log(1/(\epsilon_2-\epsilon_1)))}$ queries.
\end{theorem}

Unfortunately, their result together with the quantum tester from \cite{quantumtester} does not directly imply~\Cref{theorem: main}. When $\epsilon_2-\epsilon_1\to 0$, the quantum tester requires that either $\epsilon_1,\epsilon_2\to 0$ or $\epsilon_1, \epsilon_2\to 1/2$; while \cite{CP23}'s lower bound requires $\epsilon_1,\epsilon_2$ to be bounded away from both $0$ and $1/2$.
To further clarify why the lower bound in \cite{CP23} does not apply to our setting, we provide a brief overview of the key ideas of their proof—some of which also play important roles in our proof.
\paragraph{Chen-Patel's Proof.}
The key of~\cite{CP23}'s proof is the construction of two hard distributions $\mathcal{D}_Y$ and $\mathcal{D}_N$, which are \emph{mostly} supported on ``yes'' and ``no'' instances, respectively.
By Yao's minimax lemma, it suffices to show that any deterministic tester making a polynomial number of queries cannot distinguish between the two distributions.
Below, we will only focus on the core of the construction in~\cite{CP23}, which corresponds to the parameter regime where
\[
n=k+\ell, \qquad \ell <k^{0.99}, \qquad \tau_2=\tau_0+2^{-4\ell},\qquad \tau_1=\tau_0+2^{-2\ell}-2^{-4\ell},
\]
where $\tau_1\coloneqq 1-2\epsilon_1,\tau_2\coloneqq 1-2\epsilon_2$ are the correlations between $\mathcal{D}_Y,\mathcal{D}_N$ and the closest $k$-junta, and $\tau_0\coloneqq \mathbb{E}_{\bm{a}\sim \{\pm 1\}^{2^\ell}}[|\sum_{i=1}^{2^\ell}\bm{a}_i|/2^{\ell}]\approx \sqrt{2/\pi} \cdot 2^{-\ell/2}$.

The hard distributions $\mathcal{D}_Y$ and $\mathcal{D}_N$ are specified as follows:
\begin{itemize}
        \item $\mathcal{D}_N$: The uniform distribution over all $n$-bit Boolean functions.
        \item $\mathcal{D}_Y$: Sample $\bm{S}\sim \binom{[n]}{k}$.
Then, for each $y \in \{0,1\}^{\bm{S}}$, sample the values $\{\bm{f}(y , z)\}_{z \in \{0,1\}^{\overline{\bm{S}}}}$ uniformly at random from $\{\pm 1\}$, conditioned on the product of these values being $-1$:
\[
\prod_{z \in \{0,1\}^{\overline{\bm{S}}}} \bm{f}(y , z) = -1.
\]
\end{itemize}

A straightforward application of Hoeffding's inequality shows $\adv(\bm{g},\mathcal{J}_k)$ is concentrated around $\tau_0$ over $\bm{g}\sim \mathcal{D}_N$, hence at most $\tau_2$ with high probability.
On the other hand, by a coupling argument,~\cite{CP23} shows that
\[
\mathbb{E}_{\bm{a}}\left[\bigg|\sum_{i=1}^{2^\ell} \bm{a}_i\bigg|/2^{\ell}\right]\ge \tau_0+2^{-2\ell},
\]
where $\bm{a}$ is uniformly sampled from $\{\pm 1\}^{2^\ell}$ conditioned on $\prod_{i=1}^{2^\ell}\bm{a}_i=-1$.
Using the same reasoning, we confirm that $\adv(\bm{f},\mathcal{J}_k)$ is at least $\tau_1$ with high probability over $\bm{f}\sim \mathcal{D}_Y$.\footnote{In~\cite{CP23}, the authors mention they didn't get a tight bound on the gap $\tau_1-\tau_2$, since a loose bound was sufficient for their results.
Indeed, we believe that with some more delicate analysis, the right gap should  $\Theta(2^{-3\ell/2})$ instead of $\Theta(2^{-2\ell})$.
In any case, since we can transform the odd-parity distribution into the uniform distribution by flipping at most one bit, the gap cannot exceed $2^{-\ell}\ll \tau_0$, and such a gap does not suffice for our purpose.}

To show the two distributions are indistinguishable, Chen and Patel~\cite{CP23} observe the following: conditioned on any fixed $\bm{S}=S\in \binom{n}{k}$, an algorithm can distinguish the two distributions only if for some $y\in \{0,1\}^S$, all the points in $\{y , z\}_{z\in \{0,1\}^{\overline{S}}}$ are queried.
In particular, after querying $x^{(1)},\ldots,x^{(m)}$, the algorithm can achieve non-zero advantage only if there exists a pair of indices $1\le i<j\le m$, such that $x^{(i)}=(x^{(j)})^{\oplus {\overline{S}}}$. 
Since the probability that $x^{(i)}=(x^{(j)})^{\oplus \overline{S}}$ for random $\bm{S}\sim \binom{[n]}{k}$ is at most $1/\binom{n}{k}$, by union bound, the algorithm distinguishes two distributions with advantage at most $m^2/\binom{n}{k}$.
Therefore, any tester that distinguishes between the two distributions with constant advantage requires $\Omega\left(\sqrt{\binom{n}{k}}\right)$ queries. 

In the above construction, $\mathbb{E}_{\bm{g}\sim \mathcal{D}_N}[\adv(\bm{g},\mathcal{J}_k)]$ equals $(1-o(1)) \cdot \mathbb{E}_{\bm{f}\sim \mathcal{D}_Y}[\adv(\bm{f},\mathcal{J}_k)]$, whereas we require the former to be quadratic in the latter.
Although~\cite{CP23} extends their construction to a wider class of parameters by padding and planting biased output bits, a straightforward adaptation of their proof does not work in the parameter regime required for our quantum tester. This is because their ``yes'' distribution is too close to the ``no'' distribution.

\paragraph{An Unsuccessful Attempt.}
Let $\mathcal{D}_N$ remain the uniform distribution.
Our objective is to find a ``pseudorandom'' distribution $\mathcal{D}_Y$ with correlation at least $ 2^{-\ell/4}$ to $k$-juntas.
To achieve this, we still first sample $\bm{S}\sim \binom{[n]}{k}$, which is presumed to be the support of the nearest $k$-junta.
Then, for any fixed assignment $y\in \{0,1\}^S$, we aim to ensure that $\mathbb{E}[|\sum_{z\in \{0,1\}^{\overline{S}}} \bm{f}(y , z)|]\ge  2^{3\ell/4}$.
This allows us to apply Hoeffding's inequality to show that $\adv(f,\mathcal{J}_{\bm{S}})=\Omega(2^{-\ell/4})$ holds with high probability.
A natural way to implement this is to first sample a uniform sign $\bm{b}_y\sim \{\pm 1\}$. Then generate $\{\bm{f}(y , z)\}_{z\in \{0,1\}^{\overline{S}}}$ from an independent biased distribution towards $\bm{b}_y$, that is, for each $z\in \{0,1\}^{\overline{S}}$, let $\bm{f}(y , z)=\bm{b}_y$ with probability $1/2+ 2^{-\ell/4}$, and $\bm{f}(y , z)=-\bm{b}_y$ otherwise.

Unfortunately, this strategy is vulnerable to a simple, efficient distinguisher based on estimating influences. While a function $\bm{g} \sim \mathcal{D}_N$ satisfies $\Inf_i[\bm{g}] \ge 1/2-\Theta(2^{-n/2} \cdot \sqrt{\log n})$ for all $i \in [n]$ with high probability, a function $\bm{f} \sim \mathcal{D}_Y$ exhibits significantly lower influence:
for any $i \in \overline{\bm{S}}$, $\Inf_i[\bm{f}]  \le 1/2 -  \Theta(2^{-\ell/2})$ with high probability.
Consequently, an algorithm can distinguish the two distributions by estimating $\Inf_i[\bm{f}]$ for each $i\in [n]$ with precision $O(2^{-\ell/2})$ using $\Theta(2^{\ell})=\Theta(1/{\epsilon})$ samples, and then simply checking if the minimum of the estimations over $i\in [n]$ is below a threshold of $1/2-\Theta(2^{-\ell/2})$.

\paragraph{Our Strategy.}
The fundamental flaw in the previous construction is that it introduces local correlations between the values of $\bm{f}$ at adjacent points.
More generally, $\mathcal{D}_Y$ remains distinguishable from $\mathcal{D}_N$ if the marginal distribution of $\bm{f}$ on any Hamming ball of constant radius deviates from being uniform.
To circumvent this, we propose a construction based on error-correcting codes.
By planting the bias only on the codewords of a high-distance code, we ensure that correlations only emerge between inputs with large Hamming distance; hence, $\bm{f}$ is locally uniform.

Formally, our ``yes'' distribution $\mathcal{D}_Y$ is generated as follows:
\begin{enumerate}[noitemsep]
\item Sample a support $\bm{S} \sim \binom{[n]}{k}$.
\item Let $C \subseteq \{0,1\}^{\overline{S}}$ be an arbitrary error-correcting code with rate $R = 3/4$ and relative distance $\delta = 0.01$,
 as guaranteed by the Gilbert-Varshamov bound.
\item For each assignment $y \in \{0,1\}^S$, sample a uniform sign $\bm{b}_y \sim \{\pm 1\}$.
\item For each $z \in \{0,1\}^{\overline{S}}$, let
\[
\bm{f}(y , z)
\begin{cases} 
\coloneqq \bm{b}_y & \text{if } z \in C, \\
\sim \{\pm 1\} & \text{if } z \notin C.
\end{cases}
\]
\end{enumerate}

Since $\mathbb{E}_{\bm{f}}[|\sum_{z\in \{0,1\}^{\overline{S}}}\bm{f}(y , z)|]\ge |C|-2^{\ell/2}=\Omega(2^{3\ell/4})$,
a standard application of Hoeffding's inequality confirms that $\adv(\bm{f},\mathcal{J}_k)\ge \Omega(2^{-\ell/4})$ happens with high probability.

\paragraph{Indistinguishability.}
Next, we present the main idea behind the indistinguishability of $\mathcal{D}_Y$ and $\mathcal{D}_N$ generated as described above.
Let $\mathcal{A}^f$ be an arbitrary algorithm making $m$ queries to $f$,
and consider any path $\mathcal{P}$ in its decision tree representation.
Let $x^{(1)},x^{(2)},\ldots,x^{(m)}$ denote the set of points queried along $\mathcal{P}$.
 Under the uniform distribution $\bm{g} \sim \mathcal{D}_N$, $\mathcal{A}^{\bm{g}}$ follows $\mathcal{P}$ with probability exactly $2^{-m}$. We show that under $\bm{f} \sim \mathcal{D}_Y$, this probability remains $(1 \pm o(1))2^{-m}$ as long as $m=k^{O(\ell)}$.
The indistinguishability of $\mathcal{D}_Y$ and $\mathcal{D}_N$ follows from a standard coupling argument.
Below, we briefly discuss how to show $\mathcal{A}^{\bm{f}}$ follows $\mathcal{P}$ with probability close to $2^{-m}$ over $\bm{f}\sim \mathcal{D}_Y$.

We say $S\in\binom{[n]}{k}$ is good if for every distinct pair of query points $x^{(i)},x^{(j)}$, at least one of the following holds:
\begin{itemize}[noitemsep]
\item $x^{(i)}$ and $x^{(j)}$ are inconsistent on at least one of the coordinates in $S$,
\item  $x^{(i)}$ and $x^{(j)}$ has Hamming distance less than $\ell/100$.
\end{itemize}
We claim that conditioned on $\bm{S}=S$ for a good $S$, the queried values $\bm{f}(x^{(1)}), \dots, \bm{f}(x^{(m)})$ are independent and uniform.
To see this, fix any $y\in \{0,1\}^S$.
The ``good'' property ensures that the subcube $\{x \in \{0,1\}^n \mid x_S=y\}$ contains at most one queried point $x^{(i)}$ whose projection on $\overline{S}$ is a codeword in $C$,
because if there were two such points, they must agree on $S$ while having a Hamming distance of at least $\ell/100$, which contradicts the assumption that $S$ is good. 
Since each query hits at most one codeword per subcube,
and the signs $(\bm{b}_y)_{y\in \{0,1\}^S}$ associated with the codewords are independent and uniform,
we deduce that $\bm{f}(x^{(1)}),\ldots,\bm{f}(x^{(m)})$ are independent and uniform.
Therefore, conditioned on $\bm{S}$ being good, $\mathcal{A}^{\bm{f}}$ traverses $\mathcal{P}$ with probability exactly $2^{-m}$.

It remains to show that a random $\bm{S}\sim \binom{[n]}{k}$ is good with high probability.
For any two queried points $x^{(i)},x^{(j)}$ with Hamming distance at least $\ell/100$, the probability that they coincide on coordinates in $\bm{S}$ is at most $\binom{n-\ell/100}{k}/\binom{n}{k}=k^{-\Omega(\ell)}$.
By applying a union bound over all $m^2$ pairs of queries, we conclude that $\bm{S}$ is good with probability at least $1-m^2\cdot k^{-\Omega(\ell)}$.


\subsection{Future Directions}

An immediate question left open by our work is whether tolerant junta testing admits a quantum advantage larger than quasi-polynomial.
\begin{question}
    What is the largest possible separation between the classical and quantum query complexity of tolerant junta testing? Can we get an exponential speedup?
\end{question}

Resolving this question would require a more comprehensive understanding of tolerant junta testing in the classical setting, which currently remains out of reach.
Recall that the current best-known classical algorithm~\cite{CPS26} for tolerant junta testing has query complexity $2^{O(k^{1/3})}$, whereas the best-known classical lower bound is only quasi-polynomial~\cite{CP23}.

Another open question is whether we can obtain a quantum advantage for $(k,\eps_1,\eps_2)$-tolerant junta testing, where $\eps_1, \eps_2$ are constants bounded away from $0, 1/2$ and each other. Currently, our quasi-polynomial advantage holds only when $\eps_1$ and $\eps_2$ are $1/\poly(k)$ close to $1/2$ (and each other).

It also remains an interesting open question whether there is an efficient quantum tolerant junta tester in the parameter regime $(1-2\epsilon_1)^2\le (1-2\epsilon_2)$ that is not covered by our results (\Cref{theorem: quantum tester}).
Answering this question negatively would require new super-polynomial quantum query lower bound for tolerant junta testing.

\subsection{Organization}
The remainder of this paper is structured as follows: In \Cref{sec: preliminaries}, we provide the necessary notation and recall the standard tools in Boolean function analysis, coding theory, and concentration inequalities.
In~\Cref{sec: classical lower bound}, we establish our new classical lower bound.
Finally, in \Cref{sec: quantum tester}, we describe our quantum tester and analyze its query complexity and correctness.

\section{Preliminaries} \label{sec: preliminaries}

Throughout the paper, we use $\log$ to denote the base-$2$ logarithm, and $\ln$ to denote the natural logarithm.
For any distribution $\mathcal{D}$ over a finite space $\Omega$, we use $\mathcal{D}(x)$ to denote $\Pr_{\bm{x}\sim \mathcal{D}}[\bm{x}=x]$ for every $x\in \Omega$.
Given two distributions $\mu,\nu$ over the same space, we use $\dTV(\mu,\nu)$ to denote the \emph{total variation distance} between $\mu$ and $\nu$.
Given $S\subseteq [n]$, we use $\overline{S}\coloneqq[n]\setminus S$ to denote the complement of $S$ when $n$ is clear from the context.
Given $x\in \{0,1\}^n$, we use $x^{\oplus S}\coloneqq x\oplus 1^S0^{\overline{S}}$ to denote the string obtained by flipping all the coordinates in $S$ of $x$.

\paragraph{Boolean Functions.}
Every $n$-bit Boolean function $f:\{0,1\}^n\to \{\pm 1\}$ can be uniquely written as a polynomial in the Fourier basis: $f(x)=\sum_{S\subseteq [n]}\hat{f}_S\chi_S(x)$, where $\chi_S(x)\coloneqq (-1)^{\sum_{i\in S} x_i}$, and $\{\hat{f}_S\}_{S\subseteq [n]}$ are the Fourier coefficients of $f$.

Parseval's identity states that the sum of the squares of the Fourier coefficients of any Boolean function is exactly $1$, that is,  $\sum_{S\subseteq [n]} \hat{f}^2_S=1$.
Plancherel's theorem generalizes Parseval's identity.
It states that for any two functions $f,g:\{0,1\}^n\to \mathbb{R}$, $\mathbb{E}_{\bm{x}}[f(\bm{x})g(\bm{x})]=\sum_{S\subseteq [n]} \hat{f}_S\hat{g}_S$.

We recall the standard notion of influence for Boolean functions.
\begin{definition}[Influence]
    Let $f:\{0,1\}^n\to \{\pm 1\}$ be a Boolean function and $i\in [n]$. The influence of $i$ on $f$ is defined as
    \[
    \Inf_{i}[f]\coloneqq \Pr_{\bm{x}\sim \{0,1\}^n}[f(\bm{x})\ne f(\bm{x}^{\oplus \{i\}})]=\sum_{S\ni i} \hat{f}_S^2.
    \]
\end{definition}

In this paper, we also work with the following two variants of influence.

\begin{definition}[Level-$k$ Influence]
Let $f:\{0,1\}^n\to \{\pm 1\}$ be a Boolean function and $k,i\in [n]$. The level-$k$ influence of $i$ on $f$ is defined as
\[
    \Inf^{\le k}_{i}[f]\coloneqq\sum_{S\ni i:|S|\le k} \hat{f}_S^2.
\]
\end{definition}

\begin{definition}[Generalized Influence]
Let $f:\{0,1\}^n\to \{\pm 1\}$ be a Boolean function and $T\subseteq [n]$ be an arbitrary set.
The influence of $T$ on $f$ is defined as
\[
\Inf_T(f)\coloneqq\sum_{S:S\cap T\ne \varnothing} \hat{f}_S^2=1-\sum_{S:S\subseteq \overline{T}} \hat{f}^2_S.
\]
\end{definition}

\paragraph{Juntas.}
For any set $S\subseteq [n]$, we say a Boolean function $f:\{0,1\}^n\to \{\pm 1\}$ is an $S$-junta if it only depends on variables in $S$.
We denote the set of $S$-juntas on $n$ variables by $\mathcal{J}_S^n$.
Moreover, for any integer $k\in [n]$, we say $f$ is a $k$-junta if it depends on at most $k$ variables.
We denote the set of $k$-juntas on $n$ variables by 
$\mathcal{J}_k^n$.
Note that $\mathcal{J}_k^n=\bigcup_{S\subseteq [n]:|S|=k} \mathcal{J}_S^n$.
We will write $\mathcal{J}_S^n$ and $\mathcal{J}_k^n$ simply as $\mathcal{J}_S$ and $\mathcal{J}_k$ when $n$ is clear from the context.

For any two Boolean functions $f,g:\{0,1\}^n\to \{\pm 1\}$, we use $\dist(f,g)\coloneqq \Pr_{\bm{x}}[f(\bm{x})\ne g(\bm{x})]$ to denote the distance between $f$ and $g$.
We also define $\adv(f,g)\coloneqq \mathbb{E}_{\bm{x}}[f(\bm{x})g(\bm{x})]=1-2\dist(f,g)$ to denote the correlation between $f$ and $g$.
Let $\mathcal{F}$ be a set of $n$-bit Boolean functions.
We use $\dist(f,\mathcal{F})\coloneqq \min_{g\in \mathcal{F}} \dist(f,g)$ to denote the distance between $f$ and the closest function in $\mathcal{F}$.
Similarly, we use $\adv(f,\mathcal{F})\coloneqq \max_{g\in \mathcal{F}} \adv(f,g)=1-2\dist(f,\mathcal{F})$ to denote the maximum correlation between $f$ and any function in $\mathcal{F}$.

Given $T\subseteq [n]$, the correlation between $f$ and the closest $T$-junta can be calculated as follows.
\begin{claim}[\cite{ITW21}]\label{claim: correlation with T-juntas}
    Let $f:\{0,1\}^n\to \{\pm 1\}$ be a Boolean function and $T\subseteq [n]$. Define $f_{\avg,T}:\{0,1\}^n\to [-1,1]$ where
    \[
    f_{\avg,T}(x)\coloneqq \mathbb{E}_{\bm{y}\sim \{0,1\}^n}[f(\bm{y})\mid \bm{y}_T=\bm{x}_T].
    \]
    It holds that
    \[
    \adv(f,\mathcal{J}_T)=\mathbb{E}_{\bm{x}\sim \{0,1\}^n}[|f_{\avg,T}(\bm{x})|].
    \]
\end{claim}

The following lemma states that $\adv(f,\mathcal{J}_k)$ can be approximated by the maximum correlation between $f$ and $k$-juntas supported only on coordinates with high level-$k$ influence.
\begin{lemma}[\cite{ITW21}]\label{lemma: approximation by high infleunced coordinates}
Let $f:\{0,1\}^n\to \{\pm 1\}$ be any Boolean function and $\tau>0$ be any parameter. Let $C\coloneqq \{i\in [n]\mid \Inf_i^{\le k}[f]\ge \tau^2/k\}$. Then $|C|\le k^2/\tau^2$.
Moreover, 
\[
\max_{S\subseteq C:|S|\le k}\adv(f,\mathcal{J}_S)\ge \adv(f,\mathcal{J}_{k})-\tau.
\]
\end{lemma}

\paragraph{Quantum Query Algorithms.}
Let $f:\{0,1\}^n\to \{\pm 1\}$ be a Boolean function with Fourier representation $f(x)=\sum_{S\subseteq [n]} \hat{f}_S\chi_S(x)$.
By Parseval's theorem, one can define a natural distribution $\mathcal{S}_f$ over $2^{[n]}$ where $\mathcal{S}_f(S)=\hat{f}_S^2$. 
The following folklore states that $\mathcal{S}_f$ can be sampled by an algorithm that makes a single quantum query to $f$.
\begin{claim}[\cite{BV97}]
    There is an algorithm $\textsc{FourierSample}^f$ that makes one quantum query to $f:\{0,1\}^n\to \{\pm 1\}$, and for every $S\subseteq [n]$, it outputs $S$ with probability $\hat{f}^2_S$.
\end{claim}

\paragraph{Error-Correcting Codes.}

Given two $n$-bit strings $x,y\in \{0,1\}^n$, let $\Delta(x,y)\coloneqq \{i\mid x_i\ne y_i\}$ denote the Hamming distance between $x$ and $y$.
Moreover, the relative distance between $x$ and $y$ is defined as $\delta(x,y)\coloneqq \Delta(x,y)/n$.

\begin{definition}
    A binary code of block length $n$ is a subset $C\subseteq \{0,1\}^n$.
    The rate of $C$ is defined as $\log |C|/n$, and the relative distance of $C$ is defined as $\min_{c_1\ne c_2\in C} \delta(c_1,c_2)$.
    We say $C$ is an $(n,R,\delta)$-code if it has block length $n$, rate $R$, and relative distance $\delta$.
\end{definition}

The following textbook result states that as long as $R+H_2(\delta)\le 1$, there exists an $(n,R,\delta)$-code,
where $H_2(x)\coloneqq  x\log(1/x)+(1-x)\log(1/(1-x))$ is the binary entropy function, extended continuously to the boundaries so that $H_2(0)=H_2(1)=0$.
\begin{theorem}[Gilbert–Varshamov bound~\cite{Gilbert1952,varshamov1957}]
    Let $\delta \in (0,1/2)$ and $R \in (0,1)$ be parameters such that $R+H_2(\delta) \le 1$.
    Then there exists an $(n,R,\delta)$-code for every $n\in \mathbb{N}$.
\end{theorem}


We also use $H_2^{-1}:[0,1]\to [0,1/2]$ to denote the inverse binary entropy function, where given $x\in [0,1]$, $H^{-1}_2(x)$ is defined as the unique $y\in [0,1/2]$ such that $H_2(y)=x$.
\begin{fact}[{\cite[Theorem 2.2]{Calabro09}}]\label{fact: bound inverse binary entropy}
For every $x\in [0,1]$, it holds that $H_2^{-1}(x)\ge \frac{x}{2\log (6/x)}$.
\end{fact}

\paragraph{Concentration Bounds.}

We recall the following standard concentration inequalities.
\begin{lemma}[Hoeffding's inequality~\cite{Hoeffding1963}]
Let $\bm{X}_1,\ldots,\bm{X}_n\in [0,1]$ be independent random variables and $\bm{S}=\sum_{i=1}^n \bm{X}_i$.
Then for every $t>0$,
\[
\Pr[|\bm{S}-\mathbb{E}[\bm{S}]|\ge t]\le 2\exp(-2t^2/n).
\]
\end{lemma}
\begin{fact}[\cite{Haagerup1981}]\label{fact: expected abs sum rademacher}
    Let $n$ be even and $\bm{a}_1,\ldots,\bm{a}_n\sim \{\pm 1\}$ be $n$ independent Rademacher variables.
    Then $\mathbb{E}_{\bm{a}}[|\sum_{i=1}^n \bm{a}_i|]\le \sqrt{2n/\pi}$.
\end{fact}

We will also use the following simple inequality.
\begin{restatable}{fact}{BinomQuotient}\label{lemma: binom quotient}
    Let $0\le k,m\le n$ be integers such that $m+k\le n$. Then
    \[
    \binom{n-m}{k}\left/\binom{n}{k}\right.\le \left(\frac{n-k}{n}\right)^m.
    \]

\end{restatable}
\begin{proof}
Let $\ell \coloneqq n-k$. Then,
\[
\binom{n-m}{k}/\binom{n}{k} = \binom{n-m}{\ell-m}/\binom{n}{\ell} = \prod_{i=0}^{m-1} \frac{\ell-i}{n-i} \le (\ell/n)^{m}.\qedhere
\]
\end{proof}

\section{Classical Lower Bound}\label{sec: classical lower bound}
In this section, we demonstrate our new classical lower bound for tolerant junta testing.
\begin{restatable}{theorem}{ClassicalLowerBound}\label{theorem: classical lower bound}
Let $k\in \mathbb{N}$, $\tau_2\in (2^{-{k}^{0.99}},1/2)$, and $\tau_2<\tau_1=\tau_2^c$ and for some $c\in (\frac{\log\log (1/\tau_2)}{\log (1/\tau_2)},1)$.
Then any adaptive classical algorithm which
 \begin{itemize}
     \item accepts $f$ with probability at least $2/3$ if $f$ has correlation at least $\tau_1$ to some $k$-junta;
   \item rejects $f$ with probability at least $2/3$ if $f$ has correlation at most $\tau_2$ to every $k$-junta
 \end{itemize}
 must make $k^{\Omega(\frac{\log (1/\tau_1)}{\log (1/c)+1})}$ queries to $f$. 
 \end{restatable}

To prove the theorem, we follow the standard approach for establishing testing lower bound: Using Yao's minimax lemma, it suffices to find two hard distributions $\mathcal{D}_Y$ and $\mathcal{D}_N$---supported mostly on ``yes'' instances and ``no'' instances respectively---and prove that they are indistinguishable by any classical algorithm making a polynomial number of queries.

The rest of the section is organized as follows: In~\Cref{sec: hard distributions}, we define the hard distributions $\mathcal{D}_Y$ and $\mathcal{D}_N$ and verify that they satisfy the required correlation properties.
Then in~\Cref{sec: indistinguishability}, we show the indistinguishability of these two distributions.
Finally, in~\Cref{sec: putting everything together} we put everything together to show~\Cref{theorem: classical lower bound}.
\subsection{Hard Distributions}\label{sec: hard distributions}

Given any $k\in \mathbb{N}$, $\tau_2\in (2^{-{k}^{0.99}},1/2)$, and $\tau_2<\tau_1=\tau_2^c$ and for some $c\in (\frac{\log\log (1/\tau_2)}{\log (1/\tau_2)},1)$,
let $\ell\coloneqq 2\log (1/\tau_2)$.
We will focus only on the special case where $n=k+\ell$.
The proof for arbitrary large $n$ follows from a simple padding argument, namely, adding dummy input bits to $f$.

The two hard distributions $\mathcal{D}_N$ and $\mathcal{D}_Y$ are described as follows.

\begin{definition} \label{def: hard distributions}
Let $R\coloneqq 1-c/2+2/\ell$.
For every $S\in \binom{[n]}{k}$, let $C_{\overline{S}}\subseteq \{0,1\}^{\overline{S}}$ be a fixed arbitrary $(\ell,R, \delta)$-code with $\delta\coloneq H_2^{-1}(1-R)$.
The  distributions $\mathcal{D}_Y$ and $\mathcal{D}_N$ are defined as follows:
\begin{itemize}
        \item $\mathcal{D}_N$: The uniform distribution over all $n$-bit Boolean functions.
        \item $\mathcal{D}_Y$: Draw $\bm{S}\sim \binom{[n]}{k}$.
        For each $y\in \{0,1\}^{\bm{S}}$, draw $\bm{b}_y\sim \{\pm 1\}$.
        Then for each $z\in \{0,1\}^{\overline{\bm{S}}}$, let
        \[
        \bm{f}(y , z)\begin{cases}
            \coloneqq \bm{b}_y & z\in C_{\overline{\bm{S}}}\\
            \sim \{\pm 1\} & z\notin C_{\overline{\bm{S}}}
        \end{cases}.
        \]
\end{itemize}
\end{definition}

We first show $\mathcal{D}_N$ is mostly supported on functions with correlation at most $\tau_2$ with every $k$-junta.
\begin{claim}\label{claim: correlation with no}
Let $\bm{f}\sim \mathcal{D}_N$. With probability at least $1-\exp(-2^{k/3})$, $\adv(\bm{f},\mathcal{J}_k)\le 2^{-\ell/2}=\tau_2$.    
\end{claim}
\begin{proof}
    Fix any $S\in \binom{[n]}{k}$.
    For any assignment $y\in \{0,1\}^S$ over coordinates in $S$, define $\bm{a}_y\coloneqq 2^{-\ell}\cdot |\sum_{z\in \{0,1\}^{\overline{S}}}\bm{f}(y , z)|$.
    By~\Cref{claim: correlation with T-juntas}, it holds that $\adv(\bm{f},\mathcal{J}_S)=2^{-k}\cdot \sum_{y\in \{0,1\}^S} \bm{a}_y$.
    
    We first bound the expected value of each individual $\bm{a}_y$: Since $\{\bm{f}(y , z)\}_{z\in \{0,1\}^{\overline{S}}}$ are independent Rademacher variables, by~\Cref{fact: expected abs sum rademacher}, we have $\mathbb{E}[\bm{a}_y]\le \sqrt{2/\pi}\cdot2^{-\ell/2}$.
    Then using Hoeffding's inequality, we have $\adv(\bm{f},\mathcal{J}_S)=2^{-k} \cdot\sum_{y\in \{0,1\}^S }\bm{a}_y\ge 2^{-\ell/2}=\tau_2$ with probability $1-\exp(-2^{k/2})$.
    Finally, by applying union bound over all $S\in \binom{[n]}{k}$, we obtain
    \[
    \Pr[\adv(\bm{f},\mathcal{J}_k)\ge \tau_2]\le \sum_{S:|S|=k} \Pr[\adv(\bm{f},\mathcal{J}_S)\ge \tau_2]\le \binom{n}{k}\cdot \exp(-2^{k/2})\le \exp(-2^{k/3}).\qedhere
    \]
\end{proof}

Then we demonstrate that $\mathcal{D}_Y$ is mostly supported on functions with correlation at least $\tau_1$ with some $k$-juntas.
\begin{claim}\label{claim: correlation with yes}
Let $\bm{f}\sim \mathcal{D}_Y$. With probability at least $1-\exp(-2^{k/2})$, $\adv(\bm{f},\mathcal{J}_k)\ge \tau_1$.
\end{claim}
\begin{proof}
    Fix any $S\in \binom{[n]}{k}$. 
    Sample a uniform $\bm{f}\sim \mathcal{D}_Y$ conditioned on $\bm{S}=S$. 
    We similarly define $\bm{a}_y\coloneqq 2^{-\ell}\cdot |\sum_{z\in \{0,1\}^{\overline{S}}}\bm{f}(y , z)|$ for each assignment $y\in \{0,1\}^S$.
    By~\Cref{claim: correlation with T-juntas}, $\adv(\bm{f},\mathcal{J}_S)=2^{-k} \cdot\sum_{y\in \{0,1\}^S }\bm{a}_y$.
    However, for this time, we have
    \[
    \bm{a}_y=2^{-\ell}\cdot \left||C_{\overline{S}}|\cdot \bm{b}_y+\sum_{z\notin C_{\overline{S}}} \bm{f}(y , z)\right|.
    \]
    Using the fact that $\{\bm{f}(y , z)\}_{z\notin C_{\overline{S}}}$ are independent, an application of~\Cref{fact: expected abs sum rademacher} yields
    \[
    \mathbb{E}[\bm{a}_y]\ge |C_{\overline{S}}|/2^{\ell}-2^{-\ell/2}\ge 2^{-(1-R)\cdot \ell-1}.
    \]
Consequently, by Hoeffding's inequality, with probability at least $1-\exp(-2^{k/2})$, we have
\[
\adv(\bm{f},\mathcal{J}_S)=2^{-k}\cdot \sum_{y\in \{0,1\}^S} \bm{a}_y\ge 2^{-(1-R)\cdot\ell-2}= 2^{-c\cdot\ell/2}=\tau_2^c=\tau_1,
\]
where the second equality follows from the definition of $R\coloneqq 1-c/2+2/\ell$.
The desired claim follows by averaging over $\bm{S}\sim \binom{[n]}{k}$.
\end{proof}

\subsection{Indistinguishability}\label{sec: indistinguishability}

Next, we show that $\mathcal{D}_Y$ and $\mathcal{D}_N$ are indistinguishable by classical algorithms making a polynomial number of queries.
\begin{lemma}\label{lemma: indistinguishability}
    Let $\mathcal{D}_Y$ and $\mathcal{D}_N$ be defined as in~\Cref{def: hard distributions}.
    Then for any classical algorithm $\mathcal{A}^f$ that makes $m$ queries to $f$,
    \begin{equation}\label{eqn: indistinguishability}
    \left|\Pr_{\bm{f}\sim \mathcal{D}_Y}[\mathcal{A}^{\bm{f}}=1]-\Pr_{\bm{g}\sim \mathcal{D}_N}[\mathcal{A}^{\bm{g}}=1]\right|\le m^2\cdot k^{-\Omega(\delta\cdot \ell)}.
    \end{equation}
\end{lemma}

\begin{proof}[Proof of~\Cref{lemma: indistinguishability}]
    Consider $\mathcal{A}^f$ as a depth-$m$ decision tree.
    Without loss of generality, we can assume that every root-to-leaf path in the tree has length exactly $m$.
    Let $\mathcal{P}$ denote an arbitrary path.
    It is straightforward to observe that $\mathcal{A}^{\bm{f}}$ follows path $\mathcal{P}$ with probability exactly $2^{-m}$ when $\bm{f}$ is drawn from the uniform distribution $\mathcal{D}_N$.
    Our primary task is to show that when $\bm{f} \sim \mathcal{D}_Y$, $\mathcal{A}^{\bm{f}}$ follows $\mathcal{P}$ with probability close to $2^{-m}$.
    
    Let $x^{(1)},\ldots,x^{(m)}$ denote the points queried on $\mathcal{P}$.
    We say $S\in \binom{[n]}{k}$ is \emph{good} (w.r.t. $\mathcal{P}$) if for every $1\le i<j\le n$, either $x^{(i)}_S\ne x^{(j)}_S$, or $\Delta(x^{(i)},x^{(j)})<\delta\cdot \ell$.
    We first show that when $S$ is good, the probability that $A^{\bm{f}}$ follows $\mathcal{P}$ is exactly $2^{-m}$ over $\bm{f}\sim\mathcal{D}_Y$ conditioned on $\bm{S}=S$.
        \begin{claim}\label{claim: uniform when S is good}
        Let $r^{(1)},\ldots,r^{(m)}\in \{\pm 1\}$ be $m$ arbitrary signs. Suppose $S$ is good, then
        \[
        \Pr_{\bm{f}\sim\mathcal{D}_Y\mid \bm{S}=S}[\bm{f}(x^{(i)})=r^{(i)},\forall i\in [m]]=2^{-m}.
        \]
        In particular, $\mathcal{A}^{\bm{f}}$ follows $\mathcal{P}$ with probability exactly $2^{-m}$ over $\bm{f}\sim \mathcal{D}_Y$ conditioned on $\bm{S}=S$.
        \end{claim}
        \begin{proof}
            For every $i\in [m]$, we can write $x^{(i)}=(y^{(i)}, z^{(i)})$ where $y^{(i)}\in \{0,1\}^S,z^{(i)}\in \{0,1\}^{\overline{S}}$.
            Since $S$ is good, for any pair of distinct indices $i,j$ such that $y^{(i)}=y^{(j)}$, we must have $\Delta(z^{(i)},z^{(j)})<\delta\cdot \ell$.
            This implies either $z^{(i)}\notin C_{\overline{S}}$, or $z^{(j)}\notin C_{\overline{S}}$.
            Moreover, since the signs $(\bm{b}_y)_{y\in \{0,1\}^S}$ associated with the codewords are independent and uniform,
            we deduce that $\bm{f}(y^{(1)}, z^{(1)}),\ldots,\bm{f}(y^{(m)}, z^{(m)})$ are independent uniform bits,
            from which the desired statement follows.
        \end{proof}
        Next, we show $\bm{S}$ is good with high probability over $\bm{S}\sim \binom {[n]}{k}$.
        \begin{claim}\label{claim: most of S are good}
        With probability $1-m^2\cdot n^{-\Omega(\delta \cdot \ell)}$ over $\bm{S}\sim \binom{[n]}{k}$, $\bm{S}$ is good.
        \end{claim}
        \begin{proof}
            For any fixed $1\le i<j\le m$, let $T\coloneqq \{r\in [n]\mid x^{(i)}_r\ne x^{(j)}_r\}$.
            Observe that whenever $|T|\ge \delta\cdot \ell$,
            \[
            \Pr_{\bm{S}\sim \binom{n}{k}}[\bm{S}\cap T=\emptyset]\le \binom{n-\delta\cdot \ell}{k}\left/\binom{n}{k}\right.\le (\ell/n)^{\delta\cdot \ell}\le k^{0.01\delta\cdot \ell},
            \]
            where the second inequality follows from~\Cref{lemma: binom quotient}, and the last inequality follows from the assumption that $\ell\le k^{0.99}$.
            Therefore, 
            \[
            \Pr_{\bm{S}\sim \binom{n}{k}}[x^{(i)}_{\bm{S}}= x^{(j)}_{\bm{S}}\land \Delta(x^{(i)},x^{(j)})\ge\delta\cdot \ell]\le n^{-\Omega(\delta \cdot \ell)}.
            \]
            Finally, by union bound over all pairs $(i,j)$,
            \[
            \Pr_{\bm{S}}[\bm{S} \text{ is not good}]\le \sum_{1\le i<j\le m} \Pr_{\bm{S}}[x^{(i)}_{\bm{S}}= x^{(j)}_{\bm{S}}\land \Delta(x^{(i)},x^{(j)})\ge \delta\cdot \ell]\le  m^2\cdot k^{-\Omega(\delta\cdot \ell)}.\qedhere 
            \]
        \end{proof}
    With the above two claims in hand, we are ready to show~\eqref{eqn: indistinguishability}.
    Let $\mu_Y,\mu_N$ denote the distributions of the path which $\mathcal{A}^{\bm{f}}$ follows for $\bm{f}\sim \mathcal{D}_Y$ and $\bm{f}\sim\mathcal{D}_N$ respectively.
    Claims \ref{claim: uniform when S is good} and \ref{claim: most of S are good} together imply that for every path $\mathcal{P}$, $\mu_Y(\mathcal{P})\ge (1-m^2\cdot k^{-\Omega(\delta\cdot \ell)})\cdot 2^{-m}$.
    Hence, the total variation distance between $\mu_Y$ and $\mu_N$ is
    \[
    \dTV(\mu_Y,\mu_N)= \sum_{\mathcal{P}} \max(0,\mu_N(\mathcal{P})-\mu_Y(\mathcal{P}))\le \sum_{\mathcal{P}} \mu_N(\mathcal{P})\cdot (m^2\cdot k^{-\Omega(\delta\cdot \ell)})=m^2\cdot k^{-\Omega(\delta\cdot \ell)}.
    \]
It follows that
\[
\left|\Pr_{\bm{f}\sim \mathcal{D}_Y}[\mathcal{A}^{\bm{f}}=1]-\Pr_{\bm{g}\sim \mathcal{D}_N}[\mathcal{A}^{\bm{g}}=1]\right|\le \dTV(\mu_Y,\mu_N)\le m^2\cdot k^{-\Omega(\delta\cdot \ell)}.\qedhere
\]
\end{proof}


\subsection{Putting Everything Together}\label{sec: putting everything together}
We are now ready to prove \Cref{theorem: classical lower bound}.
\begin{proof}[Proof of~\Cref{theorem: classical lower bound}]
Let $\mathcal{D}_Y$ and $\mathcal{D}_N$ be defined as in~\Cref{def: hard distributions}.
Let $\mathcal{D}'_Y$ denote the distribution of $\bm{f}\sim \mathcal{D}_Y$ conditioned on $\adv(\bm{f},\mathcal{J}_k)\ge \tau_1$, and $\mathcal{D}'_N$ denote the distribution of $\bm{g}\sim \mathcal{D}'_N$ conditioned on $\adv(\bm{g},\mathcal{J}_k)\le \tau_2$.
Note that $\mathcal{D}'_Y$ is fully supported on ``yes'' instances, and $\mathcal{D}'_N$ is fully supported on no instances.
Moreover, \Cref{claim: correlation with yes} implies that $\dTV(\mathcal{D}_Y,\mathcal{D}'_Y)\le \exp(-2^{k/2})$, and similarly, \Cref{claim: correlation with no} implies that $\dTV(\mathcal{D}_N,\mathcal{D}'_N)\le \exp(-2^{k/3})$.
Thus, for any $m$-query algorithm $\mathcal{A}^f$, by~\Cref{lemma: indistinguishability}, we have
\begin{align*}
&\left|\Pr_{\bm{f}\sim \mathcal{D}'_Y}[\mathcal{A}^{\bm{f}}=1]-\Pr_{\bm{g}\sim \mathcal{D}'_N}[\mathcal{A}^{\bm{g}}=1]\right|\\
\le& \dTV(\mathcal{D}'_Y,\mathcal{D}_Y)+\left|\Pr_{\bm{f}\sim \mathcal{D}_Y}[\mathcal{A}^{\bm{f}}=1]-\Pr_{\bm{g}\sim \mathcal{D}_N}[\mathcal{A}^{\bm{g}}=1]\right|+\dTV(\mathcal{D}_N,\mathcal{D}'_N)\\
\le& 2m^2\cdot k^{-\Omega(\delta\cdot \ell)},\\
\le& 2m^2\cdot k^{-\Omega(\frac{\log (1/\tau_1)}{\log(6/c)})},
\end{align*}
where the last inequality follows because
\[
\delta\cdot\ell=H_2^{-1}(c/2-2/\ell)\cdot \ell\ge \Omega\!\left(\frac{c}{\log (6/c)}\cdot \log(1/\tau_2)\right)= \Omega\!\left(\frac{\log (1/\tau_1)}{\log(1/c)+1}\right).
\]

Let $\mathcal{A}^f$ be any algorithm that satisfies the conditions in the theorem statement.
It distinguishes $D'_Y$ from $D'_N$ with advantage at least $1/3$.
Therefore, it must make $k^{\Omega(\frac{\log (1/\tau_1)}{\log(1/c)+1})}$ queries.
\end{proof}

\section{Quantum Tolerant Junta Tester} \label{sec: quantum tester}

In this section, we introduce the quantum tolerant junta tester.

\begin{restatable}{theorem}{QuantumTester}\label{theorem: quantum tester}
    Let $k\in \mathbb{N}$ and $0<\tau_2<\tau_1^2<1$ be parameters. 
    Then there exists a non-adaptive quantum algorithm with query complexity $\poly(k,1/\tau)$ which
    \begin{itemize}
        \item accepts $f$ with probability at least $2/3$ if $f$ has correlation at least $\tau_1$ to some $k$-junta;
    \item rejects $f$ with probability at least $2/3$ if $f$ has correlation at most $\tau_2$ to every $k$-junta.
\end{itemize}
\end{restatable}

\begin{algorithm}[htbp] 
    \caption{Quantum tolerant junta tester } \label{alg: quantum tester}
    \begin{algorithmic}[1]
        \Require{Quantum oracle access to $f$, an integer $k$, and parameters $\epsilon_1,\epsilon_2$ with $(1-2\epsilon_1)^2>1-2\epsilon_2$ }
        \Ensure{``Yes'' or ``No''}
        \State{$\epsilon\gets ((1-2\epsilon_1)^2-(1-2\epsilon_2))/2$}
        \State $\tau\gets \epsilon/3$
        \State{$M\gets \frac{100k\log k\cdot\log (1/\tau)}{\tau^2}$}
        \State{Run $\textsc{FourierSample}^f$ $M$ times to obtain $\bm{S}^{(1)},\ldots,\bm{S}^{(M)}\sim \mathcal{S}_f$.}\label{line: sample}
        \State{$\bm{p}_j\gets 0$ for all $j\in [n]$}\label{line: block1 start}
        \For{$i\in [M]$}
            \If{$|\bm{S}^{(i)}|\le k$}
            \State $\bm{p}_j\gets \bm{p}_{j}+1/M$ for all $j\in \bm{S}^{(i)}$
            \EndIf
        \EndFor
        \State $\widetilde{\bm{C}}\gets \{j\in [n]:\bm{p}_j\ge \frac{\tau^2}{2k}\}$\label{line: block1 end}
        \State{$\bm{q}_R=0$ for all $R\subseteq \widetilde{\bm{C}}$ with $|R|\le k$}\label{line: block2 start}
        \For{$i\in [M]$}
            \State $\bm{q}_R\gets \bm{q}_R+1/M$ for all $\bm{S}^{(i)}\subseteq R\subseteq \widetilde{\bm{C}}$ with $|R|\le k$
        \EndFor\label{line: block2 end}
        \State{$\widetilde{\bm{\adv}}\gets \max_{R\subseteq \widetilde{\bm{C}}:|R|\le k} \bm{q}_R$}\label{line: approx adv}
        \If{$\widetilde{\bm{\adv}}>((1-2\epsilon_1)^2+(1-2\epsilon_2))/2$} 
            \State{\Return ``Yes''}
        \Else 
            \State{\Return ``No''}
        \EndIf
    \end{algorithmic}
\end{algorithm}

The quantum tolerant junta tester described in~\Cref{theorem: quantum tester} is implemented by~\Cref{alg: quantum tester}.
It is essentially an adaptation of the tester from~\cite{quantumtester}, tailored to the following refined trade-off between advantage and influence:

\begin{restatable}{lemma}{AdvantageEqualInfluence}   
\label{lemma: advantage equal influence}
Let $f:\{0,1\}^n \to \{\pm 1\}$ be a Boolean function and $T\subseteq[n]$ be an arbitrary set.
Then
\[
1-\Inf_{\overline{T}}(f)\le \adv(f,\mathcal{J}_T)\le \sqrt{1-\Inf_{\overline{T}}(f)}.
\]
\end{restatable}
The upper bound was shown in~\cite{quantumtester} as an intermediate step, and the lower bound was shown in~\cite{testingjuntas}.

We reprove this lemma here for completeness.
\begin{proof}
Let $f_{\avg,T}$ be as given in~\Cref{claim: correlation with T-juntas}.
It holds that $\adv(f,\mathcal{J}_T)=\mathbb{E}_{\bm{x}}[|f_{\avg,T}(\bm{x})|]$. Moreover,
\[
f_{\avg,T}(x)=\sum_{S\subseteq [n]} \hat{f}_S\mathbb{E}_{\bm{y}\sim \{0,1\}^n}[\chi_S(\bm{y})\mid \bm{y}_S=x_S]=\sum_{S\subseteq T} \hat{f}_S\chi_S(x).
\]
Using Parseval's identity, we have $\mathbb{E}_{\bm{x}}[f_{\avg,T}(\bm{x})^2] = \sum_{S\subseteq T} \hat{f}_S^2 = 1-\Inf_{\overline{T}}(f)$.
Finally, since $f_{\avg,T}(x) \in [-1,1]$ for all $x\in \{0,1\}^n$,  we get
\begin{align*}
    \mathbb{E}_{\bm{x}}[f_{\avg,T}(\bm{x})^2]
    \le \mathbb{E}_{\bm{x}}[|f_{\avg,T}(\bm{x})|]
    \le \sqrt{\mathbb{E}_{\bm{x}}[f_{\avg,T}(\bm{x})^2]}.
\end{align*}
This completes the proof.
\end{proof}


Regarding query complexity, the tester only queries $f$ during the execution of $\textsc{FourierSample}^f$ at line~\ref{line: sample}.
Consequently, it makes $M=\frac{100k\log k\cdot\log (1/\tau)}{\tau^2}$ quantum queries in total, as desired.
The remainder of this section is dedicated to proving the correctness of the algorithm.

Before presenting the formal proof, we provide a brief overview of how the algorithm proceeds. 
First, the algorithm draws $M$ random subsets $\bm{S}^{(1)},\ldots,\bm{S}^{(M)}\sim \mathcal{S}_f$.
Then from line~\ref{line: block1 start} to~\ref{line: block1 end}, the tester utilizes these samples to compute $\bm{p}_j$ as an estimate $\Inf^{\le k}_j[f]$, and identifies a set $\widetilde{\bm{C}}$ of coordinates with high estimated level-$k$ influence.

Next, from line~\ref{line: block2 start} to~\ref{line: block2 end}, the tester uses the sampled subsets to compute $\bm{q}_R$ as an estimate $1-\Inf_{\overline{R}}[f]$ for all $R\subseteq C$ with size $|R|=k$.
Finally, the algorithm computes $\tilde{\bm{\adv}}$ to approximate $\adv(f,\mathcal{J}_k)$ and compares this value with a certain threshold to determine whether to accept or reject $f$.

\begin{lemma}
Fix any $k,\epsilon_1,\epsilon_2$ with $(1-2\epsilon_1)^2>1-2\epsilon_2$.
Let $\mathcal{A}^f$ denote~\Cref{alg: quantum tester}. Then:
\begin{itemize}
\item for every $f$ with $\adv(f,\mathcal{J}_k)>1-2\epsilon_1$, with high probability $\mathcal{A}^f$ outputs ``YES''; and
\item for every $f$ with $\adv(f,\mathcal{J}_k)<1-2\epsilon_2$, with high probability $\mathcal{A}^f$ outputs ``NO''.
\end{itemize}
\end{lemma}
\begin{proof}
    Fix any Boolean function $f:\{0,1\}^n\to \{\pm 1\}$.
    Recall that~\Cref{alg: quantum tester} computes $\widetilde{\bm{\adv}}$ as an estimation of $\adv(f,\mathcal{J}_k)$ at line~\ref{line: approx adv}.
    It suffices to show that for any Boolean function $f$, with high probability $\adv^2(f,\mathcal{J}_k)-\epsilon < \widetilde{\bm{\adv}}< \adv(f,\mathcal{J}_k)+\epsilon$ for $\epsilon\coloneqq((1-2\epsilon_1)^2-(1-2\epsilon_2))/2$.

    Let $C\coloneqq \{j\in [n]\mid \Inf^{\le k}_j[f]\ge \tau^2/k\}$ denote the set of coordinates with level-$k$ influence at least $\tau^2/k$.
    We show that \Cref{alg: quantum tester} computes a superset $\widetilde{\bm{C}}\supseteq C$ with high probability at line~\ref{line: block1 end}.
    \begin{claim}
        It always holds that $|\widetilde{\bm{C}}|\le \frac{2k^2}{\tau^2}$.
        Moreover, with probability at least $1-k^{-2}$, $C\subseteq \widetilde{\bm{C}}$.
    \end{claim}
    \begin{proof}
        The first statement holds because $\sum_{j\in [n]} \bm{p}_j\le k$.
    
        For the second statement, observe that for each $j\in [n]$, $\bm{p}_j$ is the average of $M$ independent random bits with mean $\Pr_{\bm{S}\sim \mathcal{S}_f}[j\in \bm{S}\land |\bm{S}|\le k]=\Inf^{\le k}_{j}[f]$.
        By Chernoff bound, with probability $1-\exp(\log k\cdot \log \tau)$, we have $\bm{p}_j\ge \frac{\tau^2}{2k}$.
        The desired statement then follows from a union bound over all $j\in C$.
    \end{proof}

Next, we observe that for each $R\subseteq \widetilde{\bm{C}}$ with size $|R|\le k$, $\bm{q}_R$ is the average of $M$ independent random bits with mean $\Pr_{\bm{S}\sim \mathcal{S}_f}[\bm{S}\subseteq R]=1-\Inf_{\overline{R}}[f]$.
By Hoeffding's inequality, with probability at least $1-\exp(100k\log k\cdot \log \tau)$ we have $|\bm{q}_R-(1-\Inf_{\overline{R}}[f])|\le \tau$.
In which case, we can further apply~\Cref{lemma: advantage equal influence} to obtain
\[
\adv^2(f,\mathcal{J}_R)-\tau< \bm{q}_R< \adv(f,\mathcal{J}_R)+\tau.
\]
Then by applying union bound over all $R\subseteq \widetilde{\bm{C}}$ with size $|R|=k$, with probability at least $1-\exp(-\Omega(k\log k))$,
\[
\max_{R\subseteq \widetilde{\bm{C}}:|R|\le k }\adv^2(f,\mathcal{J}_R)-\tau< \widetilde{\bm{\adv}}< \max_{R\subseteq \widetilde{\bm{C}}:|R|\le k }\adv(f,\mathcal{J}_R)+\tau.
\]
When $C\subseteq \widetilde{\bm{C}}$, by~\Cref{lemma: approximation by high infleunced coordinates}, we further have
\[
\adv^2(f,\mathcal{J}_k)-\epsilon= \adv^2(f,\mathcal{J}_k)-2\tau-\tau< \widetilde{\bm{\adv}}< \adv(f,\mathcal{J}_k)+\epsilon.\qedhere
\]
\end{proof}

\DeclareUrlCommand{\Doi}{\urlstyle{sf}}
\renewcommand{\path}[1]{\small\Doi{#1}}
\renewcommand{\url}[1]{\href{#1}{\small\Doi{#1}}}
\bibliographystyle{alphaurl}
\bibliography{references}

 \end{document}